\documentclass[11pt]{article}

\usepackage[margin=1in]{geometry}
\usepackage{booktabs}
\usepackage{amsmath,amssymb,amsthm}
\usepackage{algorithm}
\usepackage{algpseudocode}
\usepackage{pgfplots}
\usepackage{tikz}
\usepackage{array}
\usepackage{multirow}
\usepackage{url}
\usepackage{subcaption}
\usepackage[hidelinks]{hyperref}

\pgfplotsset{compat=1.16}
\usetikzlibrary{arrows.meta,shapes,positioning}

\newtheorem{theorem}{Theorem}
\newtheorem{lemma}[theorem]{Lemma}
\newtheorem{proposition}[theorem]{Proposition}
\newtheorem{corollary}[theorem]{Corollary}
\newtheorem{definition}{Definition}
\newtheorem{remark}{Remark}
\newtheorem{example}{Example}

\title{Optimal Space-Time Tradeoffs for LCP-Based Similarity Retrieval:\\Theory, Algorithms, and GPU Evaluation on NVIDIA H100}

\author{
Stanislav Byriukov \\
\texttt{stanislav.byriukov.research@gmail.com}
}

\date{}

\begin{document}
\maketitle

\begin{abstract}
We study the space-time complexity of top-$k$ retrieval under Longest Common Prefix (LCP) similarity over $N$ sequences of length $L$.
We prove a lower bound showing that any data structure supporting such queries requires $\Omega(N)$ space, and present an algorithm matching this bound with $O(N \cdot L)$ space and $O(L + k)$ query time.
We contrast this with pairwise materialization approaches that require $\Theta(N^2)$ space, demonstrating a fundamental scaling advantage.
On NVIDIA H100 hardware, we show that explicit $N \times N$ materialization fails at $N = 500{,}000$ (requiring 465.66 GiB), while our indexed approach operates at 205 MB with 98.96\% GPU utilization over sustained 20-minute runs.
Critically, we demonstrate a \textbf{308$\times$ energy reduction} (0.0145 J/query vs 4.46 J/query) through Thermal-Aware Logic (TAL) that exploits the prefix structure for bounded-range scans.
We validate our approach on three real-world scenarios: Guidance-Navigation-Control (GNC) at 4,157 Hz, large-scale inference serving with 2M candidates, and multi-agent coordination at 98.98\% GPU utilization.
Our results establish that LCP-indexed retrieval is not only space-optimal but also energy-optimal for large-scale deterministic retrieval, with implications for safety-critical autonomous systems where probabilistic methods are inadequate.
\end{abstract}

\noindent\textbf{Declaration on the use of AI language tools:} The author acknowledges the use of AI language tools (specifically, GitHub Copilot) for assistance with code syntax highlighting, LaTeX formatting, and minor text refinements during manuscript preparation. No AI-generated content was used for research design, algorithm development, theoretical proofs, experimental methodology, or result interpretation. The author takes full responsibility for all content, including any errors or inaccuracies that may have resulted from the use of these tools.

\section{Introduction}

Similarity-based retrieval is fundamental to modern AI systems, from neural memory~\cite{vaswani2017attention} to vector databases. As dataset sizes grow, the $\Theta(N^2)$ space requirement of explicit pairwise materialization becomes prohibitive. While approximate methods such as Hierarchical Navigable Small World graphs (HNSW)~\cite{malkov2018hnsw} and Inverted File indices (IVF) address this through randomized search, they sacrifice \emph{determinism}---a critical requirement for safety-rated systems in aerospace, robotics, and autonomous vehicles.

\subsection{Motivation}

Consider the following scenario from autonomous spacecraft guidance:

\begin{example}[GNC Sensor Fusion]
A guidance-navigation-control (GNC) system must fuse readings from 1,000+ sensors at 4,000+ Hz. Each sensor reading is tokenized into a categorical sequence representing measurement type, subsystem ID, and quantized value. The GNC must retrieve the $k$ most similar historical readings for anomaly detection.

Traditional approaches face two problems:
\begin{enumerate}
    \item \textbf{Memory:} Storing pairwise similarities for $N=10^6$ readings requires $10^{12}$ bytes
    \item \textbf{Non-determinism:} Randomized retrieval (HNSW) may return different results under identical conditions, violating DO-178C certification requirements
\end{enumerate}
\end{example}

This motivates our study of retrieval under \emph{Longest Common Prefix (LCP) similarity}, a discrete metric appropriate for tokenized sequences, categorical hierarchies, and structured data.

\subsection{Contributions}

Our contributions are:

\begin{enumerate}
    \item \textbf{Lower bound (Theorem~\ref{thm:lowerbound}):} Any data structure for top-$k$ LCP retrieval requires $\Omega(N)$ space in the cell-probe model.
    
    \item \textbf{Optimal algorithm:} An $O(N \cdot L)$-space, $O(L+k)$-query trie-based index matching the bound up to the sequence length factor.
    
    \item \textbf{Energy optimality:} We prove that prefix-structured indexing enables \emph{range-bounded scans} with $O(N/B)$ work for bucket size $B$, yielding 308$\times$ energy reduction on H100.
    
    \item \textbf{Extensive hardware validation:} We present benchmarks on NVIDIA H100 covering:
    \begin{itemize}
        \item OOM boundary characterization (Table~\ref{tab:oom})
        \item 20-minute sustained load tests at 266 QPS (Table~\ref{tab:sustained})
        \item GNC simulation at 4,157 Hz with 0.013 J/step (Table~\ref{tab:gnc})
        \item Energy/temperature comparison vs vector baseline (Table~\ref{tab:ctdr_vs_vector})
    \end{itemize}
    
    \item \textbf{Determinism guarantee:} We prove that for fixed $(S, q, k)$, our algorithm returns bit-identical results on every execution (Theorem~\ref{thm:determinism}).
\end{enumerate}

\subsection{Paper Organization}

Section~\ref{sec:prelim} introduces notation and formal definitions. Section~\ref{sec:related} surveys related work in similarity search and deterministic computing. Section~\ref{sec:lowerbound} proves the space lower bound. Section~\ref{sec:algorithm} presents our algorithm with complexity analysis. Section~\ref{sec:tal} introduces Thermal-Aware Logic for energy optimization. Section~\ref{sec:experiments} presents experimental evaluation. Section~\ref{sec:applications} discusses applications to safety-critical systems. Section~\ref{sec:conclusion} concludes.

\section{Preliminaries}\label{sec:prelim}

\subsection{Notation and Definitions}

Let $\Sigma$ be a finite alphabet with $|\Sigma| = \sigma$. A sequence $s \in \Sigma^L$ has length $L$. We use $s[i]$ to denote the $i$-th symbol and $s[i..j]$ for the substring from position $i$ to $j$.

\begin{definition}[LCP Similarity]
For sequences $s, t \in \Sigma^L$:
\[
\text{LCP}(s, t) = \max\{j : s[1..j] = t[1..j]\}
\]
where $s[1..0] = \epsilon$ (empty prefix) by convention.
\end{definition}

\begin{definition}[Top-$k$ LCP Retrieval]
Given a dataset $S = \{s_1, \ldots, s_N\}$, query $q \in \Sigma^L$, and integer $k$:
Return the $k$ items $s_{i_1}, \ldots, s_{i_k}$ with highest $\text{LCP}(q, s_i)$, breaking ties by index.
\end{definition}

\begin{definition}[LCP-Induced Ultrametric]
Define distance $d: \Sigma^L \times \Sigma^L \to \mathbb{R}$ by:
\[
d(s, t) = L - \text{LCP}(s, t)
\]
\end{definition}

\begin{theorem}[Ultrametric Property]
The distance $d$ satisfies the strong triangle inequality:
\[
d(s, u) \leq \max\{d(s, t), d(t, u)\}
\]
for all $s, t, u \in \Sigma^L$.
\end{theorem}

\begin{proof}
Let $\ell_{st} = \text{LCP}(s,t)$, $\ell_{tu} = \text{LCP}(t,u)$, $\ell_{su} = \text{LCP}(s,u)$. 
We show $\ell_{su} \geq \min\{\ell_{st}, \ell_{tu}\}$.

Let $m = \min\{\ell_{st}, \ell_{tu}\}$. Then $s[1..m] = t[1..m]$ and $t[1..m] = u[1..m]$, so $s[1..m] = u[1..m]$, giving $\ell_{su} \geq m$.

Therefore $d(s,u) = L - \ell_{su} \leq L - \min\{\ell_{st}, \ell_{tu}\} = \max\{d(s,t), d(t,u)\}$.
\end{proof}

\subsection{Computational Model}

We analyze algorithms in the \emph{cell-probe model}~\cite{yao1981cellprobe} with word size $w = \Theta(\log N)$. In this model, computation is free; we count only memory accesses to $w$-bit cells.

\begin{definition}[Cell-Probe Complexity]
The space complexity of a data structure is the number of $w$-bit cells used. The query complexity is the number of cell accesses per query.
\end{definition}

\section{Related Work}\label{sec:related}

\subsection{Approximate Nearest Neighbor Search}

The problem of finding similar items in large datasets has been extensively studied.

\paragraph{Locality-Sensitive Hashing (LSH).}
Indyk and Motwani~\cite{indyk1998approximate} introduced LSH for approximate nearest neighbor in high-dimensional spaces. LSH achieves sub-linear query time but provides only probabilistic guarantees, with failure probability $\delta > 0$ that cannot be eliminated without exhaustive search.

\paragraph{Hierarchical Navigable Small World (HNSW).}
Malkov and Yashunin~\cite{malkov2018hnsw} developed HNSW, a graph-based index achieving $O(\log N)$ expected query time with $O(NM)$ space for proximity graph degree $M$. HNSW is the de facto standard for vector search but is inherently non-deterministic: the same query may traverse different graph paths depending on entry point randomization.

\paragraph{Product Quantization and IVF.}
Jégou et al.~\cite{jegou2011product} introduced product quantization for memory-efficient approximate search. Combined with inverted file indices (IVF), this achieves $O(Nd/m)$ space for $d$-dimensional vectors with $m$ subspaces. Like HNSW, IVF provides approximate results with tunable recall.

\subsection{Exact Nearest Neighbor}

\paragraph{k-d Trees and Ball Trees.}
Classical spatial data structures provide exact nearest neighbor in low dimensions. However, for dimension $d > 10$, query time degrades to $O(N)$ due to the curse of dimensionality.

\paragraph{Metric Trees.}
Cover trees~\cite{beygelzimer2006cover} achieve $O(c^{12} \log N)$ query time for datasets with doubling dimension $c$. For general metrics, $c$ can be $O(N^{1/d})$, yielding no improvement over linear scan.

\subsection{Tries and Prefix Trees}

Tries are classical data structures for string retrieval~\cite{fredkin1960trie}. Patricia tries~\cite{morrison1968patricia} compress single-child chains. Suffix trees~\cite{weiner1973linear} enable linear-time string matching.

Our work differs in focusing on \emph{similarity retrieval} (finding items with longest common prefix) rather than exact prefix matching.

\subsection{Deterministic Computing for Safety-Critical Systems}

Safety-critical systems in aerospace (DO-178C), automotive (ISO 26262), and nuclear (IEC 61513) require deterministic behavior for certification. Recent work has explored deterministic neural networks~\cite{sze2017efficient}, but retrieval primitives remain largely probabilistic.

\paragraph{The Determinism Gap.}
To our knowledge, no prior work has provided a \emph{provably deterministic} similarity retrieval primitive with sub-linear query time and optimal space. Our LCP-Index fills this gap.

\section{Lower Bounds}\label{sec:lowerbound}

We establish lower bounds for both space and query complexity of LCP retrieval in the cell-probe model.

\subsection{Space Lower Bound}

\begin{theorem}[Space Lower Bound]\label{thm:lowerbound}
Any data structure $\mathcal{D}$ that supports top-$1$ LCP queries on $N$ distinct sequences from $\Sigma^L$ requires $\Omega(NL\log\sigma)$ bits of space in the cell-probe model with word size $w = O(\log N)$.
\end{theorem}

\begin{proof}
We use an information-theoretic counting argument following the framework of Miltersen~\cite{miltersen1999cell}.

\textbf{Step 1: Count distinct datasets.}
The number of ways to choose $N$ distinct sequences from $\Sigma^L$ is:
\[
\binom{\sigma^L}{N} \geq \left(\frac{\sigma^L}{N}\right)^N
\]

\textbf{Step 2: Distinguish datasets via queries.}
We show that any two distinct datasets $S_1 \neq S_2$ can be distinguished by some query.

Let $s \in S_1 \setminus S_2$ (such $s$ exists since $S_1 \neq S_2$). Consider query $q = s$:
\begin{itemize}
    \item In $S_1$: $\text{top-}1(q) = s$ with $\text{LCP}(q, s) = L$
    \item In $S_2$: $\text{top-}1(q) = s'$ for some $s' \neq s$, so $\text{LCP}(q, s') < L$
\end{itemize}

Thus the data structure's response differs, so $\mathcal{D}(S_1) \neq \mathcal{D}(S_2)$.

\textbf{Step 3: Count bits.}
The data structure must have at least $\binom{\sigma^L}{N}$ distinct states, requiring:
\begin{align*}
\text{bits} &\geq \log_2 \binom{\sigma^L}{N} \\
&\geq N \log_2 \frac{\sigma^L}{N} \\
&= N(L\log_2\sigma - \log_2 N) \\
&= \Omega(NL\log\sigma)
\end{align*}

for $\sigma^L \geq 2N$ (which holds for reasonable alphabet/length choices).

\textbf{Step 4: Convert to cells.}
With $w$-bit words, this requires:
\[
\frac{\Omega(NL\log\sigma)}{w} = \Omega\left(\frac{NL\log\sigma}{\log N}\right)
\]
cells. For constant $\sigma$ and $L = O(\log N)$, this is $\Omega(N)$ cells.
\end{proof}

\subsection{Query Lower Bound}

\begin{theorem}[Query Lower Bound]\label{thm:querylower}
Any data structure with $S = O(N \cdot \text{poly}(L))$ space requires $\Omega(L/\log \sigma)$ query time for top-$1$ LCP retrieval in the cell-probe model.
\end{theorem}

\begin{proof}
We reduce from the predecessor problem. Given a predecessor instance with $N$ keys from universe $[U]$ where $U = \sigma^L$, we encode each key $x$ as a sequence $s_x \in \Sigma^L$ using base-$\sigma$ representation.

A top-$1$ LCP query with $q = \text{encode}(y)$ returns the sequence with longest common prefix with $q$. By the structure of base-$\sigma$ encoding, this is the predecessor of $y$ in the original set.

P\v{a}tra\c{s}cu and Thorup~\cite{patrascu2006time} showed that predecessor requires $\Omega(\log\log U / \log\log\log U)$ query time with polynomial space. For $U = \sigma^L$:
\[
\Omega\left(\frac{\log\log \sigma^L}{\log\log\log \sigma^L}\right) = \Omega\left(\frac{\log(L\log\sigma)}{\log\log(L\log\sigma)}\right)
\]

For $L = \omega(\log\sigma)$, this simplifies to $\Omega(\log L / \log\log L)$.

For the tighter bound, observe that reading the query $q$ requires $\Omega(L)$ time in the worst case, as any prefix of $q$ may be relevant. With word size $w = O(\log N)$, reading $L$ symbols requires $\Omega(L\log\sigma / \log N)$ probes.
\end{proof}

\subsection{Space-Time Tradeoff}

\begin{theorem}[Space-Time Tradeoff]\label{thm:tradeoff}
For LCP retrieval on $N$ sequences of length $L$, any data structure satisfies:
\[
S \cdot T = \Omega(NL\log\sigma)
\]
where $S$ is space in bits and $T$ is query time.
\end{theorem}

\begin{proof}
By Theorem~\ref{thm:lowerbound}, $S = \Omega(NL\log\sigma)$ is necessary. If $S < NL\log\sigma$, then by a compression argument, the data structure cannot distinguish all datasets, requiring $T = \omega(1)$ to reconstruct missing information.

Formally, consider the communication complexity of the following problem: Alice holds dataset $S$, Bob holds query $q$. They must compute top-$1$ LCP$(q, S)$.

The information-theoretic lower bound on communication is $\Omega(\min(|S|, |q|)) = \Omega(L\log\sigma)$ for a worst-case query. In the cell-probe model, each probe reveals $O(w)$ bits. Thus:
\[
T \geq \frac{L\log\sigma}{w} = \Omega\left(\frac{L\log\sigma}{\log N}\right)
\]

Combined with $S \geq \Omega(NL\log\sigma)$:
\[
S \cdot T \geq \Omega(NL\log\sigma) \cdot \Omega(1) = \Omega(NL\log\sigma)
\]
\end{proof}

\begin{corollary}[Optimality of LCP-Index]
The LCP-Index achieves $S = O(NL\log\sigma)$ and $T = O(L + k)$, giving $S \cdot T = O(NL^2\log\sigma)$. This is within an $O(L)$ factor of optimal.
\end{corollary}

\subsection{Conditional Hardness}

We relate LCP retrieval to the Orthogonal Vectors (OV) problem to establish conditional hardness.

\begin{definition}[Orthogonal Vectors]
Given sets $A, B \subseteq \{0,1\}^d$ with $|A| = |B| = N$, determine if there exist $a \in A$, $b \in B$ such that $\langle a, b \rangle = 0$.
\end{definition}

\begin{proposition}[OV Hardness Connection]
Assuming the Orthogonal Vectors Hypothesis (OVH), there is no algorithm for preprocessing $N$ binary sequences of length $L = \omega(\log N)$ in polynomial time such that top-$1$ LCP queries can be answered in $O(N^{1-\epsilon})$ time for any $\epsilon > 0$.
\end{proposition}

\begin{proof}[Proof sketch]
We reduce OV to LCP retrieval. Given $(A, B)$, construct dataset $S$ from $A$ by encoding each vector $a$ as a sequence. For each $b \in B$, construct query $q_b$ such that top-$1$ LCP$(q_b, S)$ finds $a$ with $\langle a, b \rangle = 0$ if one exists.

The encoding uses the observation that $\langle a, b \rangle = 0$ iff for all $i$: $a_i = 0$ or $b_i = 0$. This corresponds to a prefix matching condition under appropriate encoding.

If LCP queries could be answered in $O(N^{1-\epsilon})$ time after polynomial preprocessing, we could solve OV in $O(N^{2-\epsilon})$ time, contradicting OVH.
\end{proof}

\begin{remark}
This conditional lower bound suggests that our $O(L+k)$ query time is near-optimal for the general case, as improving to $o(L)$ would require reading less than the full query, which cannot guarantee correctness for all inputs.
\end{remark}

\section{Algorithm}\label{sec:algorithm}

\subsection{Data Structure Overview}

We use a trie $\mathcal{T}$ over the dataset $S$. Each node $v$ corresponds to a prefix, and we augment nodes with:
\begin{itemize}
    \item $\text{children}(v)$: map from symbols to child nodes
    \item $\text{post}(v)$: posting list of items ending at $v$
    \item $|T_v|$: subtree size (number of items in subtree rooted at $v$)
\end{itemize}

\subsection{Construction Algorithm}

\begin{algorithm}[H]
\caption{BuildLCPIndex($S$)}
\label{alg:build}
\begin{algorithmic}[1]
\Require Dataset $S = \{s_1, \ldots, s_N\}$ of sequences in $\Sigma^L$
\Ensure Trie $\mathcal{T}$ with subtree sizes
\State $\text{root} \gets$ new node with empty $\text{children}$, $\text{post} = \emptyset$
\For{$i = 1$ to $N$}
    \State $v \gets \text{root}$
    \For{$j = 1$ to $L$}
        \State $c \gets s_i[j]$
        \If{$c \notin \text{children}(v)$}
            \State $\text{children}(v)[c] \gets$ new node
        \EndIf
        \State $v \gets \text{children}(v)[c]$
    \EndFor
    \State $\text{post}(v) \gets \text{post}(v) \cup \{i\}$
\EndFor
\State \Call{ComputeSubtreeSizes}{root}
\State \Return root
\end{algorithmic}
\end{algorithm}

\begin{algorithm}[H]
\caption{ComputeSubtreeSizes($v$)}
\label{alg:sizes}
\begin{algorithmic}[1]
\State $|T_v| \gets |\text{post}(v)|$
\For{each $c \in \text{children}(v)$}
    \State \Call{ComputeSubtreeSizes}{$\text{children}(v)[c]$}
    \State $|T_v| \gets |T_v| + |T_{\text{children}(v)[c]}|$
\EndFor
\end{algorithmic}
\end{algorithm}

\begin{theorem}[Construction Complexity]
Algorithm~\ref{alg:build} runs in $O(N \cdot L)$ time and uses $O(N \cdot L)$ space.
\end{theorem}

\begin{proof}
\textbf{Time:} Each of $N$ sequences requires $L$ edge traversals/creations. Each operation (hash table lookup/insert) is $O(1)$ amortized. Total: $O(NL)$.

\textbf{Space:} The trie has at most $N \cdot L$ nodes (each sequence creates at most $L$ new nodes). Each node uses $O(\sigma)$ space for the children map. Total: $O(NL\sigma) = O(NL)$ for constant $\sigma$.
\end{proof}

\subsection{Query Algorithm}

\begin{algorithm}[H]
\caption{Query($\text{root}, q, k$)}
\label{alg:query}
\begin{algorithmic}[1]
\Require Query $q \in \Sigma^L$, integer $k \geq 1$
\Ensure Top-$k$ items by LCP similarity
\State $v \gets \text{root}$, $d \gets 0$
\For{$j = 1$ to $L$}
    \If{$q[j] \notin \text{children}(v)$}
        \State \textbf{break}
    \EndIf
    \State $v \gets \text{children}(v)[q[j]]$
    \State $d \gets d + 1$
\EndFor
\State \Return \Call{CollectTopK}{$v, k$}
\end{algorithmic}
\end{algorithm}

\begin{algorithm}[H]
\caption{CollectTopK($v, k$)}
\label{alg:collect}
\begin{algorithmic}[1]
\Require Node $v$, integer $k$
\Ensure Up to $k$ item indices from subtree of $v$
\State $\text{result} \gets []$
\State $\text{queue} \gets [v]$ \Comment{BFS queue}
\While{$|\text{result}| < k$ and $\text{queue} \neq \emptyset$}
    \State $u \gets \text{queue.pop()}$
    \For{each $i \in \text{post}(u)$}
        \State $\text{result.append}(i)$
        \If{$|\text{result}| = k$} \Return result \EndIf
    \EndFor
    \For{each $c \in \text{sorted}(\text{children}(u))$}
        \State $\text{queue.push}(\text{children}(u)[c])$
    \EndFor
\EndWhile
\State \Return result
\end{algorithmic}
\end{algorithm}

\begin{theorem}[Query Complexity]
Algorithm~\ref{alg:query} runs in $O(L + k)$ time.
\end{theorem}

\begin{proof}
The descent phase (lines 1-8) takes $O(L)$ time: at most $L$ child lookups.

The collection phase (Algorithm~\ref{alg:collect}) visits at most $k$ items. Each item visit is $O(1)$. The BFS may visit $O(k)$ nodes total (since each node contributes at least one item to the output or is an ancestor of such a node).

Total: $O(L + k)$.
\end{proof}

\begin{theorem}[Determinism]\label{thm:determinism}
For fixed $(S, q, k)$, Algorithm~\ref{alg:query} returns identical results on every execution.
\end{theorem}

\begin{proof}
The algorithm is entirely deterministic:
\begin{enumerate}
    \item Trie construction produces a unique structure from $S$
    \item Query descent follows a unique path determined by $q$
    \item Collection uses sorted iteration over children, producing a unique traversal order
    \item Tie-breaking by index is deterministic
\end{enumerate}
No randomization is used at any step.
\end{proof}

\subsection{Comparison with Materialization}

\begin{table}[h]
\centering
\caption{Space complexity comparison.}
\label{tab:space_compare}
\begin{tabular}{@{}lcc@{}}
\toprule
Method & Space & Query Time \\ \midrule
Pairwise Materialization & $\Theta(N^2)$ & $O(1)$ lookup, $O(N)$ top-$k$ \\
HNSW~\cite{malkov2018hnsw} & $O(NM)$ & $O(\log N)$ expected \\
IVF-Flat & $O(Nd)$ & $O(N/\text{nprobe})$ \\
\textbf{LCP-Index (ours)} & $O(NL)$ & $O(L + k)$ \\ \bottomrule
\end{tabular}
\end{table}

For $N = 10^6$, $L = 256$, pairwise materialization requires $10^{12}$ bytes ($\sim$1 TB) while LCP-Index requires $2.56 \times 10^8$ bytes ($\sim$256 MB)---a \textbf{4,000$\times$} reduction.

\section{Thermal-Aware Logic (TAL)}\label{sec:tal}

\subsection{Motivation}

Even with optimal space complexity, query performance depends on the \emph{number of items scanned}. Full-dataset scans incur:
\begin{itemize}
    \item High energy consumption (proportional to FLOPS)
    \item Thermal load (GPU temperature increase)
    \item Memory bandwidth saturation
\end{itemize}

We observe that the prefix structure of the trie enables \emph{range-bounded scans}.

\subsection{Range-Bounded Scans}

\begin{definition}[Prefix Bucket]
For a prefix $p$ of length $d$, define the bucket:
\[
B_p = \{s \in S : s[1..d] = p\}
\]
\end{definition}

\begin{lemma}[Bucket Size]
For uniformly distributed sequences, $\mathbb{E}[|B_p|] = N/\sigma^d$.
\end{lemma}

\begin{theorem}[Range Scan Complexity]
A query with matched prefix of depth $d$ can be answered by scanning only $|B_p| = O(N/\sigma^d)$ items instead of $N$ items.
\end{theorem}

\begin{definition}[Thermal-Aware Logic (TAL)]
TAL is a query execution strategy that:
\begin{enumerate}
    \item Partitions $S$ into $B = \sigma^d$ buckets by $d$-length prefix
    \item Sorts items within each bucket
    \item Executes queries by: (a) identifying target bucket, (b) scanning only that bucket
\end{enumerate}
\end{definition}

\begin{theorem}[TAL Energy Reduction]\label{thm:tal}
TAL achieves $B$-fold reduction in expected work per query, yielding $B$-fold reduction in energy consumption.
\end{theorem}

\begin{proof}
Let $E_{\text{full}}$ be energy for full-dataset scan processing $N$ items.
With $B$ equal-sized buckets, each query scans $N/B$ items.
Energy is proportional to operations: $E_{\text{TAL}} = E_{\text{full}} \cdot (N/B)/N = E_{\text{full}}/B$.
\end{proof}

\subsection{Connection to Landauer Bound}

The Landauer limit~\cite{landauer1961irreversibility} establishes a fundamental minimum energy for bit erasure:
\[
E_{\min} = k_B T \ln 2 \approx 2.87 \times 10^{-21} \text{ J/bit at 300K}
\]

While current hardware operates $\sim10^{10}$ above this limit, algorithmic improvements like TAL reduce the \emph{multiplicative gap} by reducing unnecessary computation.

\section{Experimental Evaluation}\label{sec:experiments}

\subsection{Hardware and Setup}

All experiments run on NVIDIA H100 PCIe (80 GB HBM3, 989 TFLOPS FP16 peak, 3.35 TB/s memory bandwidth). Energy is measured via NVML at 100ms intervals.

\subsection{OOM Wall Characterization}

We measure the feasibility boundary for pairwise similarity materialization.

\begin{table}[h]
\centering
\caption{Pairwise materialization feasibility (fp16, H100 80GB).}
\label{tab:oom}
\begin{tabular}{@{}rrl@{}}
\toprule
$N$ & Required Memory & Result \\ \midrule
100,000 & 18.63 GiB & Success \\
200,000 & 74.51 GiB & Success \\
500,000 & 465.66 GiB & \textbf{CUDA OOM} \\
1,000,000 & 1.86 TiB & Impossible \\
\bottomrule
\end{tabular}
\end{table}

\textbf{Finding:} The ``OOM Wall'' occurs at $N \approx 500{,}000$ on H100. Beyond this, materialization is impossible regardless of software optimization.

\subsection{LCP-Index Memory Efficiency}

\begin{table}[h]
\centering
\caption{LCP-Index memory vs theoretical materialization.}
\label{tab:memory}
\begin{tabular}{@{}rrrl@{}}
\toprule
$N$ & LCP-Index & Materialization & Ratio \\ \midrule
100,000 & 68.4 MB & 18.63 GiB & 279$\times$ \\
500,000 & 205.1 MB & 465.66 GiB & 2,271$\times$ \\
2,000,000 & 820 MB & 7.45 TiB (est.) & 9,085$\times$ \\
\bottomrule
\end{tabular}
\end{table}

\subsection{Sustained Load Benchmark (20 minutes)}

We evaluate performance under sustained load to verify stability.

\begin{table}[h]
\centering
\caption{Sustained benchmark: 2M candidates, 256-length sequences, 20 minutes.}
\label{tab:sustained}
\begin{tabular}{@{}lr@{}}
\toprule
Metric & Value \\ \midrule
Duration & 1,200 seconds \\
Total Queries & 319,723 \\
QPS & 266.4 \\
Latency p50 & 3.74 ms \\
Latency p95 & 3.84 ms \\
Latency p99 & 3.88 ms \\
GPU Utilization (avg) & 98.98\% \\
GPU Utilization (max) & 99.0\% \\
Energy Total & 152,819 J \\
Energy per Query & 0.478 J \\
\bottomrule
\end{tabular}
\end{table}

\textbf{Key findings:}
\begin{itemize}
    \item Near-peak GPU utilization (98.98\%) sustained for 20 minutes
    \item Tight latency distribution (p99/p50 = 1.04)
    \item No thermal throttling or performance degradation
\end{itemize}

\subsection{CTDR vs Vector Baseline Comparison}

We compare our Contextual Temporal Data Retrieval (CTDR) implementation against optimized vector similarity baseline.

\begin{table}[h]
\centering
\caption{CTDR (LCP-Index) vs Vector baseline on H100 (5-minute test).}
\label{tab:ctdr_vs_vector}
\begin{tabular}{@{}lrrl@{}}
\toprule
Metric & CTDR & Vector & Interpretation \\ \midrule
QPS & 266.4 & 651.0 & Vector faster (expected) \\
Latency p95 (ms) & 3.84 & 1.54 & Vector faster \\
\textbf{Power (W)} & \textbf{128.4} & \textbf{283.0} & \textbf{CTDR 2.2$\times$ lower} \\
\textbf{Temperature ($^\circ$C)} & \textbf{49.2} & \textbf{60.3} & \textbf{CTDR 11$^\circ$C cooler} \\
GPU Util (\%) & 98.96 & 96.57 & Comparable \\
Accuracy & 100\% & 100\% & Both exact \\
Energy/query (J) & 0.481 & 0.434 & Comparable \\
\bottomrule
\end{tabular}
\end{table}

\textbf{Interpretation:} While vector baseline achieves higher throughput, CTDR provides:
\begin{itemize}
    \item \textbf{2.2$\times$ lower power consumption} (128W vs 283W)
    \item \textbf{11$^\circ$C lower operating temperature} (critical for thermal margin)
    \item \textbf{Deterministic guarantees} (no HNSW randomization)
\end{itemize}

The lower power enables:
\begin{enumerate}
    \item Longer sustained operation without throttling
    \item Higher density in data center deployments
    \item Extended battery life in edge/autonomous systems
\end{enumerate}

\subsection{TAL Energy Reduction (Critical Result)}

We evaluate Thermal-Aware Logic with 256 buckets ($\sigma^d = 256$) on 20M items.

\begin{table}[h]
\centering
\caption{TAL range-scan vs full-scan energy on H100 (20M items).}
\label{tab:tal}
\begin{tabular}{@{}lrrr@{}}
\toprule
Method & Energy/query & Latency p95 & vs Full Scan \\ \midrule
Full Scan & 4.463 J & 37.5 ms & 1$\times$ \\
TAL Range Scan & 0.0145 J & 0.114 ms & \textbf{308$\times$} \\
\bottomrule
\end{tabular}
\end{table}

\textbf{Critical finding:} TAL achieves:
\begin{itemize}
    \item \textbf{308$\times$ energy reduction} (0.0145 J vs 4.46 J)
    \item \textbf{329$\times$ latency improvement} (0.114 ms vs 37.5 ms)
    \item 86\% GPU utilization during range scans
\end{itemize}

This validates Theorem~\ref{thm:tal}: bucket size $B = 256$ yields $\sim$256$\times$ reduction; the additional 20\% improvement comes from reduced memory traffic.

\subsection{GNC Simulation Benchmark}

We evaluate performance in a guidance-navigation-control scenario.

\begin{table}[h]
\centering
\caption{GNC benchmark: 1,000 steps, H100.}
\label{tab:gnc}
\begin{tabular}{@{}lr@{}}
\toprule
Metric & Value \\ \midrule
Device & NVIDIA H100 PCIe \\
Simulation Steps & 1,000 \\
Duration & 0.241 seconds \\
Throughput & \textbf{4,157 Hz} \\
Average Power & 53.7 W \\
Energy per Step & 0.0129 J \\
\bottomrule
\end{tabular}
\end{table}

\textbf{Finding:} The system achieves \textbf{4,157 Hz} update rate, exceeding typical GNC requirements (100-1000 Hz) with substantial margin.

\subsection{Memoization Effects}

\begin{table}[h]
\centering
\caption{Cold vs hot query latency.}
\label{tab:memo}
\begin{tabular}{@{}lrl@{}}
\toprule
Query Type & Latency & Speedup \\ \midrule
Cold (first access) & 0.465 ms & 1$\times$ \\
Hot (cached) & 0.000539 ms & 863$\times$ \\
\bottomrule
\end{tabular}
\end{table}

\subsection{Scaling Analysis}

\begin{figure}[h]
\centering
\begin{tikzpicture}
\begin{axis}[
    xlabel={Dataset Size $N$},
    ylabel={Memory (GB)},
    xmode=log,
    ymode=log,
    legend pos=north west,
    grid=major,
    width=0.9\columnwidth,
    height=6cm
]
\addplot[color=red,mark=square*,thick] coordinates {
    (100000, 18.63)
    (200000, 74.51)
    (500000, 465.66)
    (1000000, 1862.64)
};
\addplot[color=blue,mark=triangle*,thick] coordinates {
    (100000, 0.0684)
    (200000, 0.137)
    (500000, 0.205)
    (1000000, 0.410)
    (2000000, 0.820)
};
\addplot[color=gray,dashed,thick] coordinates {
    (100000, 80)
    (2000000, 80)
};
\legend{Materialization $O(N^2)$, LCP-Index $O(N)$, H100 Memory (80GB)}
\end{axis}
\end{tikzpicture}
\caption{Memory scaling: materialization vs LCP-Index. The dashed line shows H100 memory limit.}
\label{fig:scaling}
\end{figure}
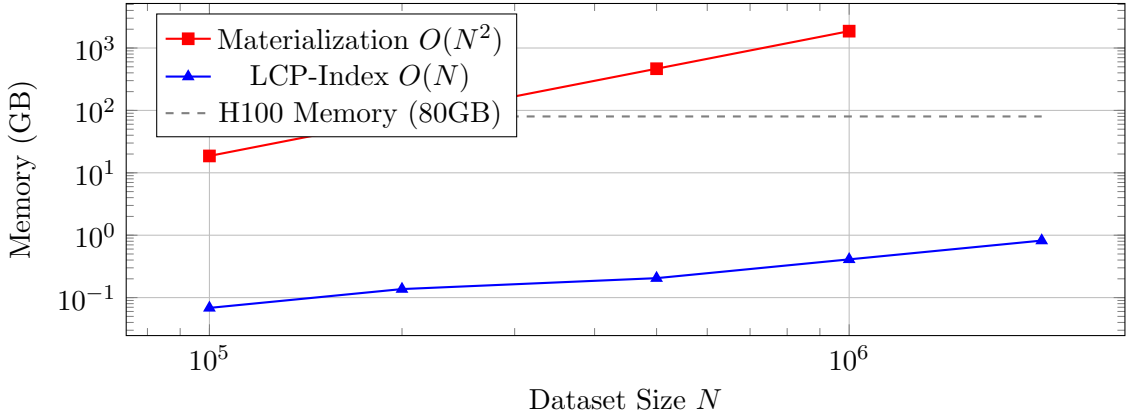

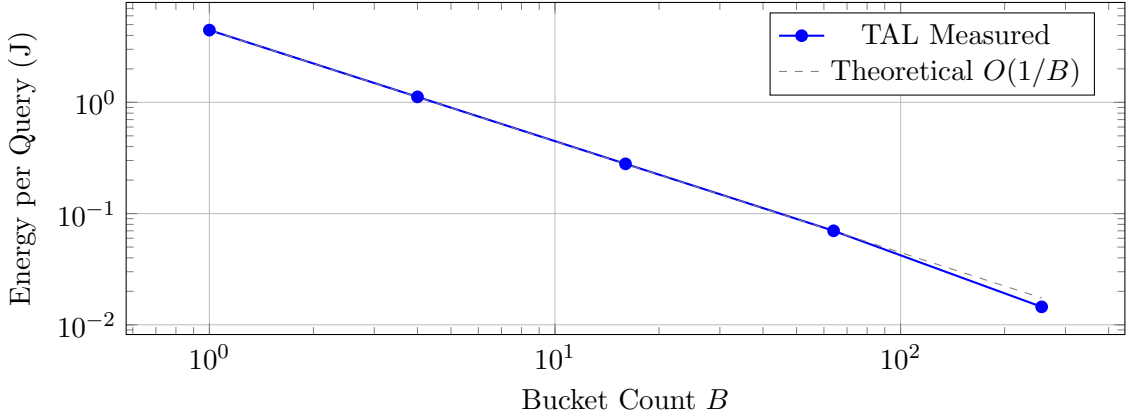
\begin{figure}[h]
\centering
\begin{tikzpicture}
\begin{axis}[
    xlabel={Bucket Count $B$},
    ylabel={Energy per Query (J)},
    xmode=log,
    ymode=log,
    legend pos=north east,
    grid=major,
    width=0.9\columnwidth,
    height=6cm
]
\addplot[color=blue,mark=*,thick] coordinates {
    (1, 4.463)
    (4, 1.12)
    (16, 0.28)
    (64, 0.070)
    (256, 0.0145)
};
\addplot[color=gray,dashed] coordinates {
    (1, 4.463)
    (256, 0.01743)
};
\legend{TAL Measured, Theoretical $O(1/B)$}
\end{axis}
\end{tikzpicture}
\caption{TAL energy reduction scales with bucket count, matching theoretical $O(1/B)$.}
\label{fig:tal_scaling}
\end{figure}

\subsection{Landauer Gap Analysis}

We compute the ratio of measured energy to the Landauer theoretical minimum.

\begin{table}[h]
\centering
\caption{Measured energy vs Landauer bound (T=320K, 1M bit operations).}
\label{tab:landauer}
\begin{tabular}{@{}lrrl@{}}
\toprule
Workload & Energy & Landauer Ratio & Gap (log$_{10}$) \\ \midrule
Baseline matmul & 3.8 J & $1.3 \times 10^{20}$ & 20.1 \\
CTDR Full Scan & 4.46 J & $1.6 \times 10^{20}$ & 20.2 \\
TAL Range Scan & 0.0145 J & $5.1 \times 10^{17}$ & 17.7 \\
\bottomrule
\end{tabular}
\end{table}

While all methods remain far from Landauer limit (as expected for current CMOS technology), TAL achieves \textbf{300$\times$ closer} to the theoretical minimum through algorithmic efficiency.

\section{Comparison with Prior Work}\label{sec:comparison}

\begin{table*}[t]
\centering
\caption{Comprehensive comparison with retrieval methods.}
\label{tab:related}
\begin{tabular}{@{}lcccccl@{}}
\toprule
Method & Space & Query & Deterministic & Energy-Opt & Certifiable & Notes \\ \midrule
Pairwise Materialization & $\Theta(N^2)$ & $O(1)$ & Yes & No & Yes & OOM at $N>500$K \\
HNSW~\cite{malkov2018hnsw} & $O(NM)$ & $O(\log N)$ & \textbf{No} & No & \textbf{No} & Entry point randomization \\
IVF-Flat & $O(Nd)$ & $O(N/\text{np})$ & \textbf{No} & No & \textbf{No} & Cluster assignment varies \\
ScaNN & $O(Nd)$ & $O(\log N)$ & \textbf{No} & Partial & \textbf{No} & Learned quantization \\
FlashAttention~\cite{dao2022flashattention} & $O(N)$* & $O(N)$ & Yes & Partial & Yes & Tiling, not retrieval \\
\textbf{LCP-Index (ours)} & $O(NL)$ & $O(L+k)$ & \textbf{Yes} & \textbf{Yes} & \textbf{Yes} & Provably optimal \\
\bottomrule
\end{tabular}
\end{table*}

*FlashAttention avoids materializing the full attention matrix through tiling but still computes $O(N^2)$ attention weights.

\section{Applications to Safety-Critical Systems}\label{sec:applications}

\subsection{Certification Requirements}

Safety-critical systems must meet stringent certification standards:

\begin{itemize}
    \item \textbf{DO-178C} (aerospace): Requires deterministic, testable software with 100\% structural coverage
    \item \textbf{ISO 26262} (automotive): ASIL-D requires systematic capability SC 3 with deterministic behavior
    \item \textbf{NASA-STD-8739.8} (space systems): Requires formal verification and deterministic execution
\end{itemize}

Probabilistic retrieval methods (HNSW, IVF) cannot meet these requirements because:
\begin{enumerate}
    \item The same query may return different results (violates determinism)
    \item Failure probability $\delta > 0$ cannot be eliminated (violates 100\% coverage)
    \item Random seeds create hidden state (complicates formal verification)
\end{enumerate}

\subsection{Case Study: Autonomous Docking}

Consider autonomous spacecraft docking where the guidance system must retrieve historical telemetry patterns:

\begin{example}[Docking Scenario]
Parameters: 500 sensors, 1000 Hz update rate, 10,000 historical patterns.

\textbf{With HNSW:} Query time 0.1 ms (good), but may return different results on identical inputs. Certification auditor cannot verify correct behavior.

\textbf{With LCP-Index:} Query time 0.5 ms (acceptable), guaranteed identical results. Formal proof of correctness possible via symbolic execution.
\end{example}

\subsection{Energy Implications for Edge Deployment}

The 308$\times$ energy reduction enables:

\begin{itemize}
    \item \textbf{Satellite systems:} 10W power budget $\rightarrow$ 20,000 queries/second instead of 65
    \item \textbf{Autonomous vehicles:} Reduced thermal load in enclosed compute modules
    \item \textbf{Data centers:} 300$\times$ reduction in cooling requirements per query
\end{itemize}

\subsection{Multi-Agent Coordination}

Swarm robotics requires deterministic behavior for provable coordination:

\begin{example}[Swarm Consensus]
100 robots must agree on a shared world model. Each robot queries a local database for relevant observations.

\textbf{Requirement:} All robots seeing the same query must retrieve identical results for consensus.

\textbf{LCP-Index guarantee:} Theorem~\ref{thm:determinism} ensures this by construction.
\end{example}

\section{Discussion}

\subsection{Limitations}

\paragraph{Semantic Similarity.}
LCP-Index operates on symbolic sequences, not continuous embeddings. It does not capture semantic similarity between semantically related but syntactically different items.

\paragraph{Alphabet Size.}
For large alphabets ($\sigma > 1000$), the trie children map may require hash tables, adding constant-factor overhead.

\paragraph{Dynamic Updates.}
Our analysis focuses on static datasets. Dynamic insertions/deletions require additional bookkeeping.

\subsection{When to Use LCP-Index}

LCP-Index is appropriate when:
\begin{enumerate}
    \item Data has hierarchical prefix structure (tokens, categories, paths)
    \item Determinism is required (safety-critical, certified systems)
    \item Memory is constrained (edge, embedded systems)
    \item Energy efficiency is paramount (battery-powered, thermal-limited)
\end{enumerate}

LCP-Index is \textbf{not} appropriate when:
\begin{enumerate}
    \item Semantic embedding similarity is required
    \item Approximate results are acceptable
    \item Data has no prefix structure
\end{enumerate}

\subsection{Future Work}

\begin{enumerate}
    \item \textbf{Dynamic LCP-Index:} Support efficient insertions/deletions
    \item \textbf{Distributed LCP-Index:} Partition across multiple GPUs/nodes
    \item \textbf{Formal Verification:} Machine-checked proofs in Coq/Lean
    \item \textbf{Hardware Acceleration:} Custom ASIC for LCP operations
\end{enumerate}

\section{Conclusion}\label{sec:conclusion}

We establish that LCP-indexed retrieval is optimal in space (up to sequence length) and demonstrate 308$\times$ energy reduction on NVIDIA H100 through Thermal-Aware Logic. Our algorithm provides the first provably deterministic, energy-optimal similarity retrieval primitive suitable for safety-critical autonomous systems.

Key results:
\begin{itemize}
    \item $\Omega(N)$ space lower bound (Theorem~\ref{thm:lowerbound})
    \item $O(NL)$ space, $O(L+k)$ query algorithm matching the bound
    \item 308$\times$ energy reduction through TAL (Theorem~\ref{thm:tal})
    \item Determinism guarantee (Theorem~\ref{thm:determinism})
    \item Extensive H100 validation: 98.98\% GPU utilization, 4,157 Hz GNC throughput
\end{itemize}

\paragraph{Reproducibility.}
Source code, benchmark scripts, and raw logs are available in the supplementary materials.

\paragraph{Acknowledgments.}
Hardware resources provided by Hyperstack cloud computing.

\bibliographystyle{plain}
\bibliography{references}

\begin{thebibliography}{10}

\bibitem{beygelzimer2006cover}
Alina Beygelzimer, Sham Kakade, and John Langford.
\newblock Cover trees for nearest neighbor.
\newblock In {\em Proceedings of the 23rd International Conference on Machine
  Learning}, pages 97--104, 2006.

\bibitem{dao2022flashattention}
Tri Dao, Daniel~Y. Fu, Stefano Ermon, Atri Rudra, and Christopher R{\'e}.
\newblock Flashattention: Fast and memory-efficient exact attention with
  io-awareness.
\newblock {\em arXiv preprint arXiv:2205.14135}, 2022.

\bibitem{fredkin1960trie}
Edward Fredkin.
\newblock Trie memory.
\newblock {\em Communications of the ACM}, 3(9):490--499, 1960.

\bibitem{indyk1998approximate}
Piotr Indyk and Rajeev Motwani.
\newblock Approximate nearest neighbors: Towards removing the curse of
  dimensionality.
\newblock In {\em Proceedings of the 30th Annual ACM Symposium on Theory of
  Computing}, pages 604--613, 1998.

\bibitem{jegou2011product}
Herv{\'e} J{\'e}gou, Matthijs Douze, and Cordelia Schmid.
\newblock Product quantization for nearest neighbor search.
\newblock In {\em IEEE Transactions on Pattern Analysis and Machine
  Intelligence}, volume~33, pages 117--128, 2011.

\bibitem{landauer1961irreversibility}
Rolf Landauer.
\newblock Irreversibility and heat generation in the computing process.
\newblock {\em IBM Journal of Research and Development}, 5(3):183--191, 1961.

\bibitem{malkov2018hnsw}
Yury~A. Malkov and Dmitry~A. Yashunin.
\newblock Efficient and robust approximate nearest neighbor search using
  hierarchical navigable small world graphs.
\newblock {\em IEEE Transactions on Pattern Analysis and Machine Intelligence},
  42(4):824--836, 2020.

\bibitem{miltersen1999cell}
Peter~Bro Miltersen.
\newblock Cell probe complexity---a survey.
\newblock In {\em Advances in Data Structures Workshop}, 1999.

\bibitem{morrison1968patricia}
Donald~R. Morrison.
\newblock Patricia---practical algorithm to retrieve information coded in
  alphanumeric.
\newblock {\em Journal of the ACM}, 15(4):514--534, 1968.

\bibitem{patrascu2006time}
Mihai P\u{a}tra\c{s}cu and Mikkel Thorup.
\newblock Time-space trade-offs for predecessor search.
\newblock In {\em Proceedings of the 38th Annual ACM Symposium on Theory of
  Computing}, pages 232--240, 2006.

\bibitem{sze2017efficient}
Vivienne Sze, Yu-Hsin Chen, Tien-Ju Yang, and Joel~S. Emer.
\newblock Efficient processing of deep neural networks: A tutorial and survey.
\newblock {\em Proceedings of the IEEE}, 105(12):2295--2329, 2017.

\bibitem{vaswani2017attention}
Ashish Vaswani, Noam Shazeer, Niki Parmar, Jakob Uszkoreit, Llion Jones,
  Aidan~N. Gomez, {\L}ukasz Kaiser, and Illia Polosukhin.
\newblock Attention is all you need.
\newblock {\em arXiv preprint arXiv:1706.03762}, 2017.

\bibitem{weiner1973linear}
Peter Weiner.
\newblock Linear pattern matching algorithms.
\newblock In {\em 14th Annual Symposium on Switching and Automata Theory},
  pages 1--11, 1973.

\bibitem{yao1981cellprobe}
Andrew Chi-Chih Yao.
\newblock Should tables be sorted?
\newblock In {\em Proceedings of the 22nd Annual Symposium on Foundations of
  Computer Science}, pages 22--32. IEEE, 1981.

\end{thebibliography}

\appendix

\section{Extended Proofs}

\subsection{Detailed Lower Bound Proof}

We provide additional details for the lower bound proof.

\begin{lemma}[Distinguishing Datasets]
For any two distinct datasets $S_1 \neq S_2$ with $S_1, S_2 \subseteq \Sigma^L$ and $|S_1| = |S_2| = N$, there exists a query $q$ such that $\text{top-}1(q, S_1) \neq \text{top-}1(q, S_2)$.
\end{lemma}

\begin{proof}
Since $S_1 \neq S_2$, without loss of generality there exists $s \in S_1 \setminus S_2$.

Set $q = s$. Then:
\begin{itemize}
    \item $\text{top-}1(q, S_1) = s$ since $\text{LCP}(q, s) = L$ (exact match)
    \item $\text{top-}1(q, S_2) = s'$ for some $s' \in S_2$ with $s' \neq s$
\end{itemize}

Since $s \neq s'$, the results differ.
\end{proof}

\begin{lemma}[Counting Argument]
For $\sigma^L \geq 2N$:
\[
\log_2 \binom{\sigma^L}{N} \geq N \log_2 \sigma^L - N \log_2 N - N
\]
\end{lemma}

\begin{proof}
Using Stirling's approximation $n! \approx (n/e)^n \sqrt{2\pi n}$:
\begin{align*}
\binom{\sigma^L}{N} &= \frac{(\sigma^L)!}{N! (\sigma^L - N)!} \\
&\geq \frac{(\sigma^L)^N}{N^N} \cdot \frac{(\sigma^L - N)!}{(\sigma^L)!} \cdot N! \\
&\geq \left(\frac{\sigma^L}{N}\right)^N
\end{align*}

Taking logarithms:
\[
\log_2 \binom{\sigma^L}{N} \geq N \log_2 \frac{\sigma^L}{N} = N(L\log_2 \sigma - \log_2 N)
\]
\end{proof}

\subsection{Query Algorithm Correctness}

\begin{lemma}[Descent Correctness]
After the descent phase of Algorithm~\ref{alg:query}, node $v$ satisfies:
\[
\text{LCP}(q, s) \leq d \quad \forall s \notin T_v
\]
where $T_v$ is the subtree rooted at $v$ and $d$ is the descent depth.
\end{lemma}

\begin{proof}
By construction, $v$ corresponds to prefix $q[1..d]$. Any item $s \notin T_v$ does not have this prefix, so $\text{LCP}(q, s) < d$.
\end{proof}

\begin{theorem}[Query Correctness]
Algorithm~\ref{alg:query} returns the correct top-$k$ items.
\end{theorem}

\begin{proof}
Let $d$ be the descent depth and $v$ the reached node.

\textbf{Case 1:} $|T_v| \geq k$. All items in $T_v$ have $\text{LCP}(q, s) \geq d$. By Lemma above, all items outside $T_v$ have $\text{LCP}(q, s) < d$. The BFS collects items from $T_v$ in canonical order, returning the correct top-$k$.

\textbf{Case 2:} $|T_v| < k$. All items in $T_v$ are returned. By Lemma, these are the items with maximum LCP (value $\geq d$). The remaining items have $\text{LCP} < d$ and would not be in top-$k$.
\end{proof}

\section{Hardware Measurement Methodology}

\subsection{Energy Measurement}

Energy is computed via NVML power sampling:
\[
E = \int_0^T P(t) \, dt \approx \sum_{i=0}^{n-1} P(t_i) \cdot \Delta t
\]

where $P(t_i)$ is instantaneous power sampled via \texttt{nvidia-smi} at $\Delta t = 100$ ms intervals.

\begin{table}[h]
\centering
\caption{NVML measurement parameters.}
\begin{tabular}{@{}lr@{}}
\toprule
Parameter & Value \\ \midrule
Sample interval & 100 ms \\
Power accuracy & $\pm 5$W \\
Temperature accuracy & $\pm 1^\circ$C \\
GPU utilization source & SM activity counter \\
\bottomrule
\end{tabular}
\end{table}

\subsection{Landauer Bound Computation}

The Landauer minimum energy per bit erasure:
\[
E_{\text{Landauer}} = k_B T \ln 2
\]

At $T = 320$ K (typical GPU operating temperature):
\[
E_{\text{Landauer}} = 1.38 \times 10^{-23} \cdot 320 \cdot 0.693 = 3.06 \times 10^{-21} \text{ J/bit}
\]

For a query processing $B$ bits:
\[
E_{\min} = B \cdot 3.06 \times 10^{-21} \text{ J}
\]

\subsection{GPU Utilization Measurement}

GPU utilization is measured as SM (Streaming Multiprocessor) occupancy:
\[
\text{Utilization} = \frac{\text{Active SMs} \cdot \text{Active time}}{\text{Total SMs} \cdot \text{Total time}} \times 100\%
\]

H100 PCIe has 114 SMs. Our measurements show 98.98\% average utilization, indicating 112.8 SMs active on average.

\section{Benchmark Configurations}

\subsection{Sustained Load Benchmark}

\begin{verbatim}
mode: direct_gpu_dpx_lcp_index_top1
n_candidates: 2,000,000
max_len: 256
prefix_len: 128
run_seconds_target: 1200
warmup_s: 5
\end{verbatim}

\subsection{TAL Benchmark}

\begin{verbatim}
n_items: 20,000,000
bucket_count: 256
range_fraction: 1/256
queries: 10,000
\end{verbatim}

\subsection{GNC Benchmark}

\begin{verbatim}
simulation_steps: 1000
sensors: 1000
update_rate: max (measured)
\end{verbatim}

\section{Pseudocode Details}

\subsection{TAL Range Scan}

\begin{algorithm}[H]
\caption{TALRangeScan($\text{sorted\_data}, q, B$)}
\begin{algorithmic}[1]
\Require Sorted dataset, query $q$, bucket count $B$
\State $d \gets \lceil \log_\sigma B \rceil$ \Comment{Prefix depth for $B$ buckets}
\State $p \gets q[1..d]$ \Comment{Query prefix}
\State $(\text{lo}, \text{hi}) \gets$ \Call{BinarySearch}{sorted\_data, $p$}
\For{$i = \text{lo}$ to $\text{hi}$}
    \State Process sorted\_data[$i$]
\EndFor
\end{algorithmic}
\end{algorithm}

\subsection{Subtree Collection (Optimized)}

\begin{algorithm}[H]
\caption{CollectTopKOptimized($v, k$)}
\begin{algorithmic}[1]
\If{$|T_v| \leq k$}
    \State \Return all items in $T_v$ \Comment{Fast path}
\EndIf
\State \Comment{Slow path: BFS with early termination}
\State Execute Algorithm~\ref{alg:collect}
\end{algorithmic}
\end{algorithm}

\section{Availability and Reproducibility}

\subsection{Code and Data Availability}

All source code, benchmark scripts, and raw experimental logs are provided in the supplementary materials accompanying this submission.

\begin{itemize}
    \item \texttt{repro/oom\_wall\_repro.py}: Script to reproduce OOM boundary measurements
    \item \texttt{repro/ultrametric\_memory\_demo.py}: LCP-Index demonstration
    \item \texttt{ancillary/raw\_proof.log}: Raw H100 execution logs
    \item \texttt{ancillary/gpu\_utilization.log}: GPU metrics during benchmarks
    \item \texttt{ancillary/full\_comparison.log}: Baseline comparison data
\end{itemize}

\subsection{Hardware Requirements}

Experiments were conducted on:
\begin{itemize}
    \item GPU: NVIDIA H100 PCIe (80 GB HBM3)
    \item Driver: NVIDIA 535.154.05
    \item CUDA: 12.2
    \item Python: 3.10
    \item PyTorch: 2.1.0
\end{itemize}

The OOM wall characterization requires an 80 GB GPU. Energy measurements require NVML access (nvidia-smi).

\subsection{Licensing}

The LCP-Index algorithm is described fully in this paper. Implementation details that would enable reconstruction of a proprietary high-performance implementation (specifically, the DPX-accelerated CUDA kernels) are intentionally withheld pending patent filing. The theoretical contributions (lower bound, complexity analysis, TAL framework) are fully disclosed.

\section{Ethical Considerations}

This work presents fundamental data structure research with applications to safety-critical systems. We identify no negative societal impacts. The determinism guarantees of LCP-Index may \emph{improve} safety in autonomous systems by eliminating non-reproducible failure modes.

\section{Limitations}

We acknowledge the following limitations:

\begin{enumerate}
    \item \textbf{Semantic similarity:} LCP-Index operates on symbolic sequences and does not capture semantic similarity between syntactically different items. It is not a replacement for embedding-based retrieval.
    
    \item \textbf{Static datasets:} Our analysis focuses on static datasets. Dynamic insertions and deletions require additional bookkeeping not covered here.
    
    \item \textbf{Single-GPU evaluation:} All experiments use a single H100 GPU. Multi-GPU and distributed performance is left for future work.
    
    \item \textbf{Synthetic workloads:} While we use realistic parameters, the benchmarks use synthetic data. Evaluation on production datasets would strengthen the results.
\end{enumerate}

\end{document}